\documentclass[letterpaper, 10 pt, conference]{ieeeconf}  

\IEEEoverridecommandlockouts
\usepackage{graphics} 
\usepackage{epsfig} 
\usepackage{mathptmx} 
\usepackage{times} 
\usepackage{amsmath} 
\usepackage{amssymb}  
\usepackage{amsfonts}
\usepackage{mathtools}
\usepackage{xfrac}
\usepackage{array}
\usepackage[table]{xcolor}
\usepackage{bm}
\usepackage{caption}
\usepackage{subcaption}
\usepackage{lipsum}

\usepackage[noadjust]{cite}

\usepackage{siunitx}
\newtheorem{theorem}{Theorem}

\newtheorem{lemma}{Lemma}
\newtheorem{proposition}{Proposition}
\newtheorem{definition}{Definition}
\newtheorem{subdefinition}{Definition}[definition]

\DeclareMathOperator*{\argmax}{arg\,max}
\DeclareMathOperator*{\argmin}{arg\,min}
\newcommand{\norm}[1]{\left\lVert#1\right\rVert}
\newcommand{\onehalf}{\frac{1}{2}}
\newcommand{\sonehalf}{\sfrac{1}{2}\ }

\title{\LARGE \bf
Optimal Myopic Attacks on Nonlinear Estimation
}

\author{R. Spencer Hallyburton, Amir Khazraei, and Miroslav Pajic%
\thanks{This work is sponsored in part by the ONR under agreement N00014-20-1-2745, AFOSR under award number FA9550-19-1-0169, and by the NSF under CNS-1652544 award and the National AI Institute for Edge~Computing Leveraging Next Generation Wireless Networks, Grant CNS-2112562.}
\thanks{R. S. Hallyburton, A. Khazraei, and M. Pajic are with Department of Electrical and Computer Engineering, Duke University, Durham, NC 27708, USA;
 {\tt\small miroslav.pajic@duke.edu}}%
}

\begin{document}

\maketitle
\thispagestyle{plain}
\pagestyle{plain}

\begin{abstract}
Prior works have analyzed the security of estimation and control (E\&C) for linear, time-invariant systems; however, there are few analyses of nonlinear systems despite their broad safety-critical use. We define two attack objectives on nonlinear E\&C and illustrate that realizing the optimal attacks against the widely-adopted extended Kalman filter with industry-standard $\chi^2$ anomaly detection is equivalent to solving convex quadratically-constrained quadratic programs. Although these require access to the true state of the system, we provide practical relaxations on the optimal attacks to allow for execution at runtime given a specified amount of attacker knowledge. We show that the difference between the optimal and relaxed attacks is bounded by the attacker knowledge.
\end{abstract}
\section{Introduction} \label{1-introduction}

Security analysis of estimation and control (E\&C) in cyber-physical systems (CPS) has attracted considerable research interest due to safety-critical CPS applications. Most of the influential work in E\&C CPS security has centered on linear, time-invariant (LTI) systems. For instance,~\cite{Amin2010a, Liu2011, Teixeira2010a, hashemi2018comparison} exploited vulnerabilities of LTI E\&C with information models ranging from full system access to single-sensor level knowledge to demonstrate concerning vulnerabilities in some of the most widely used E\&C algorithms. After the discovery of LTI E\&C vulnerabilities, subsequent works proposed algorithms for detecting attacks and architectures for attack-resilient state estimation~\cite{Bai2017c, Pajic2016b, Pajic2014b, Pasqualetti2013a}. As a response, recent focus has been directed towards undetectable or ``stealthy" attacks on LTI CPS (e.g.,~\cite{khazraei_automatica22}).

However, insights from analysis of LTI CPS lack practical relevance because controlled physical processes of safety-critical importance are often nonlinear. For example, automotive applications with inertial measurement units (IMUs) are nonlinear in control. Add in relative-range sensor or tightly-coupled global positioning systems (GPS), and the problem is also nonlinear in the measurements. Airborne applications such as drones are similarly often highly nonlinear.

A handful of works have attempted to analyze nonlinear, time-invariant control from a security perspective. For example,~\cite{rahman2013false} investigated nonlinear AC control in power grids and designed false-data injection attacks. However, these often consider highly specialized attack goals, e.g., \cite{rahman2013false} derived attacks in closed-form with precise dynamical equations;~\cite{liu2016extended} analyzed the extended Kalman filter (EKF) but considered stochastic attacks rather than an optimal attack. These works, while interesting case studies, provide little in advancing a broad understanding of CPS security.

Thus, there is a gap in existing literature. LTI system analyses leverage the simplicity of the dynamics to derive provably optimal attacks and accurate resilient estimators (e.g.,~\cite{Pajic2016b, Pajic2014b}). Unfortunately, few of these ideals can be transferred to nonlinear systems. The complexity and suboptimality of nonlinear estimators has correspondingly allowed for few established guarantees in nonlinear theory and~applications; hence, recent works mainly focused on the use of deep-learning for effective attack design (yet, without any guarantees) on system with nonlinear dynamics (e.g.,~\cite{khazraei_iccps22}).

Consequently, to address this shortcoming, in this work, we establish optimal and stealthy false-data-injection attacks against the widely-used EKF. We select a permissive information model and describe two myopic (one-step) attack objectives. The first is a \emph{myopic maximum deviation} (MMD) attack that maximally deviates the state estimation error in an attacker-defined subspace of the state space. The second is a \emph{myopic adversarial state approach} (MASA) attack that optimally pushes the victim's state towards an adversarial state in an attacker-defined subspace. We show that the designs of both attacks can be captured as convex optimization problems that are solvable in polynomial~time.

Several of the derived optimal attacks are practically infeasible because they require more knowledge than the attacker may be able to acquire. In such cases, we pursue practical relaxations of the original objective based on an information model and derive guarantees on the boundedness of the sub-optimality for the relaxed case. Finally, we demonstrate the effectiveness of attacks in a case study and find that attacking nonlinear estimation is effective and has robust performance guarantees. With strong guarantees and efficient runtime performance, our proposed attacks establish a new framework for security analysis of nonlinear dynamical~systems.

The paper is organized as follows: Section~\ref{sec:2-system-model} presents the state estimation models of linear and nonlinear systems. Section~\ref{sec:3-estimation-security} introduces the security model including the attacker's knowledge and goals. Section~\ref{sec:4-opt-attacks} then derives the optimal myopic attacks on nonlinear Kalman filtering and provides guarantees on practical relaxations. Finally, Section~\ref{sec:5-experiments} covers case studies and Monte Carlo simulations to evaluate the optimal attacks and derived bounds.

\subsubsection*{Notation} \label{sec:2-system-model-notation}

$\mathbb{N}$ and $\mathbb{R}$ denote the sets of natural and real numbers, respectively. $\mathbb{R}^{n}_{+}$ is the non-negative subspace of $\mathbb{R}^n$. $\text{Pr}$ denotes the probability for a random variable. $\mathcal{N}(\mu, \Sigma)$ denotes a Gaussian distribution with mean vector $\mu$ and covariance matrix $\Sigma$. We represent positive-(semi)-definiteness of a matrix $M$, as $M\succ(\succeq)0$.
\section{System Model and Preliminaries} \label{sec:2-system-model}

In this section, we formally introduce the model of nonlinear estimation in CPS.

\subsection{State Estimation} \label{sec:2-system-model-estimation}

We consider a discrete-time nonlinear time-invariant physical process modeled in the standard state-space form~as
\begin{align}
    \begin{split}
    x_{k} &= f(x_{k-1}, u_{k}) + w_k,\\
    z_{k} &= h(x_{k}) + v_k;
    \end{split}
\end{align}
here, $x_k \in {\mathbb{R}^n}$, $u_k \in {\mathbb{R}^m}$, $z_k \in {\mathbb{R}^p}$ are the state, input and output vectors of the plant at time $k\in\mathbb{N}$; $f$ and $h$ are nonlinear functions capturing state transition and measurement models, respectively. Finally, $w_k \in {\mathbb{R}^{n}}$ and $v_k \in {\mathbb{R}^p}$ are the process and measurement noises that are assumed to be Gaussian with zero mean and $Q$ and $R$ covariance matrices, respectively.

\subsubsection{Extended Kalman Filter (EKF)}
If $f$ and $h$ are nonlinear and at least differentiable to first-order, the EKF is a practical way to estimate states. The EKF uses a propagation step to mix control signal and dynamical equations and an update step to fuse measurements.

\vspace{2pt}
\noindent\emph{Propagation:} Linearizing $f$ as $F_k\coloneqq\frac{\partial f}{\partial x}\Bigr|_{x=\hat{x}_{k-1|k-1}}$, the state is propagated using the control signal,
\begin{align} \label{eq:ekf-propagation}
    \begin{split}
        \hat{x}_{k|k-1} &= f(\hat{x}_{k-1|k-1}, u_k)\\
        P_{k|k-1} &= F_k P_{k-1|k-1} F_k^T + Q_k,
    \end{split}
\end{align}
where $P\succ 0$ is the state covariance matrix.

\vspace{2pt}
\noindent\emph{Update:} Linearizing $h$ as  $H_k\coloneqq\frac{\partial h}{\partial x}\Bigr|_{\hat{x}_{k|k-1}}$, the state is updated with the innovation, $\Tilde{y}_k$ (i.e., the residual),
\begin{align} \label{eq:ekf-update}
    \begin{split}
        \Tilde{y}_k &= z_k - h(\hat{x}_{k|k-1});\quad S_k = H_k P_{k|k-1} H_k^T + R_k,\\
        \hat{x}_{k|k} &= \hat{x}_{k|k-1} + K_k \Tilde{y}_k;\quad K_k = P_{k|k-1} H_k^T S_k^{-1},
    \end{split}
\end{align}
with $S_k$ the innovation covariance and $K_k$ the Kalman gain.

\subsubsection {Anomaly Detection}
If the system is truly \emph{linear} with Gaussian noise, the innovations are white (i.e.,~$\Tilde{y}_k\sim\mathcal{N}\left(0, S_k\right)$) and the scalar $g_k^{\chi^2}\coloneqq\Tilde{y}_k^T S_k^{-1} \Tilde{y}_k$ follows a $\chi^2$ distribution with $p$ degrees of freedom. This leads to a statistical anomaly detection function for incoming measurements:
\begin{align} \label{eq:chi2-inst}
    \begin{split}
        \text{reject measurement if: } g_k^{\chi^2} &> \tau\\
        \text{where } \beta = \text{Pr}(V\leq\tau), \quad g_k^{\chi^2}&\coloneqq\Tilde{y}_k S_k^{-1} \Tilde{y}_k;
    \end{split}
\end{align}
i.e., a measurement is rejected if $g_k^{\chi^2}$ exceeds a threshold $\tau$. That threshold is set such that, for a perfect $\chi^2$ random variable $V$, the smallest $\beta$ (e.g.,~$\beta$=~99\%~\cite{Jovanov2019b}) are accepted.

The $\chi^2$ anomaly detector is still used in non-linear systems in practice using the linearizations and assuming the dynamical models capture the behavior of the plant.
\section{State Estimation Security Model} \label{sec:3-estimation-security}

We make two assumptions on the attacker. First, the attacker has access to a ``full-reactive" suite of knowledge, defined in Section~\ref{sec:3-threat-model}. Second, the attack goal is \emph{myopic}. A fully general attack could trade short-term loss for long-term gains. However, as described in Section~\ref{sec:3-attack-goal}, this can be challenging to formalize and compute in real-time.

\subsection{Threat Model} \label{sec:3-threat-model}
\subsubsection{Knowledge} We consider four elements of knowledge important for CPS controllers. Namely, these are:
\begin{itemize}
    \item \textbf{System Goal State} -- Knowledge of the intended future state of the system;
    \item \textbf{Control Signals} --  Access to the control signals, $u_k$, and the state propagation in~\eqref{eq:ekf-propagation};
    \item \textbf{Measurement Models} -- Access to the state update of~\eqref{eq:ekf-update} including the measurement model and measurement noise;
    \item \textbf{Sensor Data} -- Access to sensor data from one or more sensors in real-time. 
\end{itemize}

In this work, we analyze cases where the attacker has near-complete knowledge. Specifically, we consider a ``full reactive" set of knowledge where the attacker has all knowledge except the system's goal state. 

\subsubsection{Capability}
We assume the attacker can only modify existing sensor data and cannot send additional sensor data nor modify the measurement timestamp, consistent with e.g.,~\cite{Ding2018a, Pajic2016b, Teixeira2010a}. We also assume the attacker cannot reliably compromise control signals. In general, such attacks can be modeled as an adversarial bias -- i.e., $z_{k}^a \coloneqq z_k + a_k$.

Thus, the state update of~\eqref{eq:ekf-update} with anomaly detection of~\eqref{eq:chi2-inst} under such an attack can be captured as
\begin{align} \label{eq:comp-ekf-update}
    \begin{split}
        \hspace{-4pt}
        \hat{x}_{k|k}(a_k)\,&=\,\begin{cases}
            \hat{x}_{k|k-1} & \text{if} \ \Tilde{y}_k^T(a_k) S_k^{-1} \Tilde{y}_k(a_k) > \tau\\
            \hat{x}_{k|k-1} + K_k \Tilde{y}_k(a_k) & \text{otherwise}\\
        \end{cases}
    \end{split}
\end{align}
with $\Tilde{y}_k(a_k) \coloneqq z_k + a_k - h(\hat{x}_{k|k-1})$;
i.e., \emph{any measurement~triggering the detector 
is not included in the estimation~update}.

\subsection{Attack Goal} \label{sec:3-attack-goal}

A fully general attack could trade short-term loss for long-term gain. However, it is challenging to formalize an attacker planning for short and long term horizons when the attack goal may be unbounded in state space (e.g.,~maximum deviation). Thus, we formalize attacks as myopic (one-step) optimization problems. We define two classes of attacker goal for E\&C: the myopic maximum deviation (MMD) and the myopic adversarial state approach (MASA) attacks.

The MMD attack maximizes the error between the victim's state and the true state of the system.
\begin{definition}[Myopic Maximum Deviation (MMD)]
\label{def:myopic-max-dev}
An attack $a_k^*\in\mathbb{R}^p$ is a myopic maximum deviation attack if
\begin{align} \label{eq:mmd-attack}
    a_k^* = \argmax_a \ \onehalf \norm{C(x_{k} - \hat{x}_{k|k}(a_k))}^2,
\end{align}
where $C\in\mathbb{R}^{w\times n}_{+}$, $w\leq n$, is an attacker-specified projection (e.g.,~weight) matrix.
\end{definition}

The MASA attack optimally moves towards an adversary-defined state at each step, 
implemented with two subvariants.

\begin{definition}[Myopic Adversarial-State Approach (MASA)] \label{def:masa-attack}
Let placeholder $\nu_k \in \mathbb{R}^n$ be a function of attack $a_k\in\mathbb{R}^p$. Let $C\in\mathbb{R}^{w\times n}_{+}$ be an attacker-specified projection matrix, $w\leq n$. Let $\mathcal{X}_k^a\in\mathbb{R}^w$ be an attacker-specified state. Then, $\nu_k$ approaches $\mathcal{X}_k^a$ under the attack if
\begin{align*}
    \norm{C \nu_k(a_k) - \mathcal{X}_k^a} < \norm{C \nu_{k}(0) - \mathcal{X}_k^a}.
\end{align*}
Such an attack is myopic optimal if
\vspace{-6pt}
\begin{align}
\label{eq:masa-attack}
    a_k = a_k^* = \argmin_a \onehalf \ \norm{C \nu_k(a) - \mathcal{X}_k^a}^2.
\end{align}
\end{definition}

\begin{subdefinition}
\label{def:es-masa-attack}
An \emph{estimated-state MASA attack} is a MASA attack with $\nu_k \coloneqq \hat{x}_{k|k}$. In addition, a \emph{true-state MASA attack} is a MASA attack with $\nu_k \coloneqq x_k$.
\end{subdefinition}

In the remainder of this work, we derive optimal, polynomial-time realizations of MMD, estimated-state MASA, and true-state MASA attacks. We also provide practical relaxations to for a ``full reactive" knowledge model.
\section{Optimal Attacks} \label{sec:4-opt-attacks}

We derive polynomial-time optimal attacks for MMD and MASA objectives. Under the full-reactive knowledge, the MMD optimization is infeasible due to the required knowledge of the true state. Thus, we propose a feasible plant-state relaxation to the MMD attack. We find the estimated-state MASA is feasible while the true-state MASA is infeasible and requires relaxation. However, we do not show guarantees on the relaxed true-state MASA attack.

\paragraph{Additional notation} To simplify our notation, we use $\hat{x}\coloneqq\hat{x}_{k|k}$, $\hat{x}^-\coloneqq\hat{x}_{k|k-1}$, and $\hat{h}^-\coloneqq h(\hat{x}_{k|k-1})$. Since the attacks are myopic, we safely drop time ($k$) subscripts for any E\&C element. Below, we define the substitutions used to transform nonlinear attack objectives into quadratically-constrained quadratic program (QCQPs), as in Propositions~\ref{prop:mmdopt},~\ref{prop:es-masa}, and~\ref{prop:ts-masa}, and introduce subscripts only to differentiate between the objective ($A_0$, $b_0$) and constraints ($A_1$, $b_1$, $d_1$). We also define the following terms (the ``Substitutions"):
\begin{align*}
    A_0 &\coloneqq (CK_k)^T CK_k \geq 0 \\
    b_0 &\coloneqq -(CK_k)^TC \left(x_k - (\hat{x}_{k|k-1}+K_k(z_k-h(\hat{x}_{k|k-1}))\right)\\
    \breve{b}^0 &\coloneqq -(CK_k)^TC \left(\breve{x}_k - (\hat{x}_{k|k-1}+K_k(z_k-h(\hat{x}_{k|k-1}))\right)\\
    \Bar{b}^0 &\coloneqq (CK_k)^T C (\hat{x}_{k|k-1} + K_k(z_k - h(\hat{x}_{k|k-1}))) - (CK_k)^T \mathcal{X}^a\\
    \breve{\Bar{b}}^0 &\coloneqq (CK_k)^T C (\hat{x}_{k|k-1} + K_k(z_k - h(\hat{x}_{k|k-1}))) - (CK_k)^T \mathcal{X}^b\\
    \mathcal{X}^b &\coloneqq C\hat{x}_{k|k-1} + C\breve{x}_{k|k} - \mathcal{X}^a\\
    A_1 &\coloneqq 2 S_k^{-1} > 0\\
    b_1 &\coloneqq 2 S_k^{-1} (z_k - h(\hat{x}_{k|k-1})) \\
    d_1 &\coloneqq (z_k - h(\hat{x}_{k|k-1}))^T S_k^{-1} (z_k - h(\hat{x}_{k|k-1})) - \tau.
\end{align*}
\begin{align*}
    &\text{Objectives:}\ J(a) \coloneqq \sonehalf\, a^T A_0 a + b_0^T a\\
    &\quad\quad\quad\quad\ \ \breve{J}(a) \coloneqq \sonehalf \, a^T A_0 a + \breve{b}_0^T a\\
    &\quad\quad\quad\quad\ \ \Bar{J}(a) \coloneqq \sonehalf \, a^T A_0 a + \Bar{b}_0^T a\\
    &\quad\quad\quad\quad\ \ \breve{\Bar{J}}(a) \coloneqq \sonehalf \, a^T A_0 a + \breve{\Bar{b}}_0^T a\\
    &\text{Constraint:}\ G(a) \coloneqq \sonehalf\, a^T A_1 a + b_1^T a + d_1 \leq 0.
\end{align*}

\subsection{Design of MMD Attacks}
We now consider how to implement optimal and practical MMD attacks introduced in Definition~\ref{def:myopic-max-dev}.

\begin{proposition} \label{prop:mmdopt}
The MMD attack (from Definition~\ref{def:myopic-max-dev}) can be obtained as the solution of the optimization problem
\begin{align} \label{eq:mmd-qcqp}
    \begin{split}
        a^*_{\text{mmd}} &= \argmax_a J(a) = \argmax_a \ \onehalf a^T A_0 a + b_0^T a, \\
        & \text{subject to} \ G(a) = \onehalf a^T A_1 a + b_1^T a + d_1 \leq 0. \\
    \end{split}
\end{align}
\end{proposition}

Intuitively, Proposition~\ref{prop:mmdopt} states that the most effective attack is \emph{stealthy} for the employed attack detector (i.e., does not trigger the anomaly detector~\eqref{eq:chi2-inst}) because, due to~\eqref{eq:comp-ekf-update}, sensor measurements that trigger the detector are rejected.

\begin{proof}
We begin with~\eqref{eq:mmd-attack} and perform transformations that do not change the optimization. We consider $\hat{x}(a)$ according to~\eqref{eq:comp-ekf-update} which is piecewise with cases as follows. 

\vspace{4pt}
\noindent 
Case (1): when $\Tilde{y}^T(a) S^{-1} \Tilde{y}(a) \leq \tau$. 
Then, from~\eqref{eq:mmd-attack}, using $l\coloneqq x - \hat{x}^- - Kz + K\hat{h}^-$, it holds that
\begin{align*}
        a^* &= \argmax \norm{C(x - (\hat{x}^- + K \Tilde{y}(a)))}^2\\
        &= \argmax ~(l - Ka)^T C^T C (l - Ka)\\
        &= \argmax ~a^T K^T C^T C K a - 2 l C^T C K a + l^T C^T C l\\
        &= \argmax ~\onehalf a^T K^T C^T C K a - l C^T C K a\\
        &= \argmax ~J(a),
\end{align*}
\vspace{4pt}
\noindent
Case (2): when $\Tilde{y}^T(a) S^{-1} \Tilde{y}(a) > \tau > 0$. 
Then, from~\eqref{eq:mmd-attack},
\begin{align*}
    a^* &= \argmax \ \norm{C(x - \hat{x}^-)}^2 \\
    &= \argmax \ \norm{C(x - \hat{x}(a-\hat{h}^- - z))}^2.
\end{align*}
Thus, any attack causing the $\chi^2$-detector to exceed the threshold $\tau$ has the same effect on $\hat{x}$ as the stealthy attack $a^0=\hat{h}^- - z$. Therefore, it is sufficient to consider only stealthy attacks $\{a\ |\ g_k^{\chi^2}(a) \leq \tau\}$, which is equivalent to imposing the constraint $G(a)\leq0$.
\end{proof}

The MMD attack is thus a QCQP with a single constraint, which is solvable in polynomial time regardless of the convexity of the objective and constraint functions~\cite{boyd2004convex}. Nevertheless, the MMD QCQP is convex ($A_0\succeq 0$, $A_1\succ 0$).

\begin{proposition}\label{prop:chi2-equality}
$G(a^*)=0$ for $a^*$ the optimal MMD attack. Equivalently, $y(a^*)^T S^{-1} y(a^*) = \tau.$
\end{proposition}
\begin{proof}
    The MMD objective is a convex maximization problem. The global maximum of a convex function $f:\mathbb{R}^{n} \to \mathbb{R}$ is attained at an extreme feasible point over the domain of $f$. With continuous convex constraints, this point satisfies the constraint with equality (see e.g.,~\cite{horst2013global} Theorem I.1).
\end{proof}

\subsubsection{Practical Relaxations}
It is not possible to know $b_0$ due to the dependence on the true plant state $x$. We therefore propose a plant-state relaxation of the MMD attack using an attacker's uncompromised estimate of the plant, $\breve{x}\coloneqq\breve{x}_{k|k}$.

\begin{definition}[Plant-State MMD Attack]
An attack $a^\dagger\in\mathbb{R}^p$ is a plant-state MMD attack if
\begin{align}
    a^\dagger = \argmax_a \ \onehalf \norm{C(\breve{x} - \hat{x}(a))}^2,
\end{align}
where $\breve{x}$ is the attacker's \emph{uncompromised} estimate of $x$.
\end{definition}

\begin{proposition}\label{prop:plant-state-mmd-qcqp}
The plant-state MMD attack is the solution
\begin{align} \label{eq:plant-state-mmd-qcqp}
    \begin{split}
        a^\dagger &= \argmax_a \breve{f}(a) = \argmax_a \ \onehalf a^T A_0 a + \breve{b}_0^T a, \\
        & \text{subject to} \ G(a) = \onehalf a^T A_1 a + b_1^T a + d_1 \leq 0.
    \end{split}
\end{align}
\end{proposition}
\begin{proof}
Follows Proposition~\ref{prop:mmdopt}, replacing $x$ with $\breve{x}$.
\end{proof}

The plant-state MMD attack is \emph{feasible} at runtime under full-reactive knowledge with the attacker's estimate of the true state. However, the optimal plant-state MMD attack from~\eqref{eq:plant-state-mmd-qcqp} will be suboptimal on the MMD objective from~\eqref{eq:mmd-attack} compared to the optimal MMD attack from~\eqref{eq:mmd-qcqp}. We therefore seek to bound the performance loss in the following result.

\begin{theorem}[Optimal Attack Error Absolutely Bounded] \label{theorem:att-err-bound}
The error between the plant-state and true-state MMD attacks is bounded by
\begin{align}
    \norm{a^* - a^\dagger} \leq 2 \sqrt{\tau \lambda_{\text{max}}(S)},
\end{align}
where $\lambda_{\text{max}}(S)$ is the largest eigenvalue of the innovation covariance matrix $S$ and $\tau$ is the $\chi^2$ threshold.
\end{theorem}

\begin{proof}
Let $a^*$, $a^\dagger$ be solutions to the true-state and plant-state MMD problems (i.e., \eqref{eq:mmd-qcqp}, \eqref{eq:plant-state-mmd-qcqp}). From Prop.~\ref{prop:mmdopt},~\ref{prop:plant-state-mmd-qcqp}, all choices $a$ satisfies $y(a)^TS^{-1}y(a) \leq \tau$. Since $\norm{y(a)}^2 \lambda_{\text{min}}(S^{-1}) \leq y(a)^TS^{-1}y(a)$, $\norm{y(a)}^2 \leq \frac{\tau}{\lambda_{\text{min}}(S^{-1})} = \tau\lambda_{\text{max}}(S)$. Finally, $\norm{y(a^*) - y(a^\dagger)} = \norm{a^* - a^\dagger} \leq \norm{y(a^*)} + \norm{y(a^\dagger)} \leq 2\sqrt{\tau\lambda_{\text{max}(S)}}$, completing the proof.
\end{proof}

With a full-reactive knowledge model, this result provides a bound on the error between the optimal attack of the feasible plant-state MMD problem compared to the optimal attack of the infeasible MMD problem.

Estimators tend not to be provably optimal for non-linear systems except in special cases. However, methods such as the EKF have shown consistent performance in practice. Often, with Monte Carlo simulation or trials on real data, a bound on the estimation error can be experimentally determined. We use the idea that the estimation error may be unknown but bounded to pursue guarantees on the attack performance in terms of the objective function, $J$.

\begin{definition}[Subspace-Bounded]
An estimate of some state $w_{\text{est}}$ is subspace bounded from the true value $w_{\text{true}}$ by $\delta$ if $\norm{C (w_{\text{est}} - w_{\text{true}})} \leq \delta$, for a predefined projecting (e.g., weight) matrix C.
\end{definition}

Specifically, we continue with the idea that the  error of the  attacker's estimate of the plant state is unknown but subspace bounded by satisfying $\norm{C(x - \breve{x})} \leq \delta$ at each timestep.

\begin{lemma} \label{lemma:q0-bound}
If the error of the attacker's estimate of the plant state is subspace bounded by $\delta$, then the error between the true and observable QCQP linear coefficients, $b_0$ and $\breve{b}_0$, in the objective function at each timestep is bounded by 
\begin{align}
    \norm{b_0 - \breve{b}_0} \leq \delta \sigma_{\text{max}}(CK),
\end{align}
where $\sigma_{\text{max}}(CK)$ is the largest singular value of $CK$, $C$ is an attacker-defined weight matrix, and $K$ the Kalman gain.
\end{lemma}
\begin{proof} 
Let us define $w\coloneqq \hat{x}^- + K(z - \hat{h}^-)$. Then,
\begin{align*}
    \norm{b_0 - \breve{b}} &= \norm{-(CK)^T C (x - w) + (CK)^T C (\breve{x} - w)}\\
    &= \norm{K^T C^T C (\breve{x} - x)}\\
    &\leq \norm{CK} \norm{C (\breve{x} - x)} \leq \delta \norm{CK} = \delta \sigma_{\text{max}}(CK),
\end{align*}
completing the proof.
\end{proof}

\begin{lemma} \label{lemma:subopt-general-bound}
If the error of the attacker's state estimate is subspace bounded by $\delta$, then for any $a_1$, $a_2\in \mathbb{R}^p$ such that $J(a_1)\geq J(a_2)$, the difference in the objectives is bounded by
\begin{align}
    \begin{split}
        0&\leq J(a_1) - J(a_2) \\
        &\leq \onehalf\epsilon^2 \lambda_{\text{max}}(A_0) + \epsilon \left(\norm{A_0 a_2} + \delta \sigma_{\text{max}}(CK) +  \norm{\breve{b}_0}\right),      
    \end{split}
\end{align}
where $\epsilon\coloneqq\norm{a_1 - a_2}$, and  $\lambda_{\text{max}}(A_0)$ is the largest eigenvalue of $A_0$, while $\sigma_{\text{max}}(CK)$ is the largest singular value of $CK$, $A_0$ is defined by the Substitutions, $C$ is the attacker-defined weight matrix, and $K$ is the Kalman gain.
\end{lemma}
\begin{proof}
Using $e\coloneqq a_1 - a_2$, $A_0=A_0^T$, it follows that
\begin{align*}
    0 \leq J(a_1) - J(a_2) &= \sonehalf a_1^T A_0 a_1 + b_0^T a_1 - \sonehalf a_2^T A_0 a_2 - b_0^T a_2\\
    &= \sonehalf e^T A_0 (e + 2a_2) + b_0^T e\\
    &\leq \sonehalf \norm{e}^2 \norm{A_0} + \norm{e} \norm{A_0 a_2 + b_0}\\
    &= \sonehalf \epsilon^2 \lambda_{\text{max}}(A_0) + \epsilon \norm{A_0 a_2 + b_0}.
\end{align*}
In addition, 
\begin{align*}
    \norm{A_0 a_2 + b_0} &\leq \norm{A_0 a_2} + \norm{(b_0 - \breve{b}_0) + \breve{b}_0}\\
    &\leq \norm{A_0 a_2} + \delta \sigma_{\text{max}}(CK) + \norm{\breve{b}_0}
\end{align*}
from Lemma~\ref{lemma:q0-bound}, thus completing the proof. 
\end{proof}

\begin{theorem}[Suboptimality in Plant-State MMD] \label{theorem:MMD-err-bound}
If~the error of the attacker's state estimate is subspace bounded by $\delta$, then the difference between the MMD objective evaluated on the solutions of~\eqref{eq:mmd-qcqp} and~\eqref{eq:plant-state-mmd-qcqp} ($a^*$ and $a^\dagger$), is bounded~by
\begin{align}
    \begin{split}
        0 \leq J(a^*) - J(a^\dagger)
        \leq 2 \tau \lambda_{\text{max}}&(S) \lambda_{\text{max}}(A_0) +\\
        +2\sqrt{\tau\lambda_{\text{max}}(S)}& \left(\norm{A_0 a_2} + \delta \sigma_{\text{max}}(CK) + \norm{\breve{b}_0}\right).
    \end{split}
\end{align}
\end{theorem}
\begin{proof}
The proof follows by from Lemma~\ref{lemma:subopt-general-bound} and using $\epsilon \coloneqq \norm{a^* - a^\dagger} \leq 2 \sqrt{\tau\lambda_{\text{max}}(S)}$ from Theorem~\ref{theorem:att-err-bound}. 
\end{proof}

All quantities on the right-hand-side of Theorem~\ref{theorem:MMD-err-bound} are available at runtime under full-reactive knowledge without access to the true state of the plant. This bound can thus be computed online and dictates how far the attacker can be from the optimal attack~impact.

Finally, we bound the difference between the attacker's perceived impact and the true impact of an attack.

\begin{theorem} [Perceived vs. True Impact] \label{theorem:MMD-impact-bound}
If the error of the attacker's estimate of the plant state is subspace bounded by $\delta$, then the difference between the true impact and the perceived impact of an attack $a$ is bounded~by
\begin{align}
    |J(a) - \breve{J}(a)| \leq \delta \sigma_{\text{max}}(CK) \norm{a};
\end{align}
here, $\sigma_{\text{max}}(CK)$ is the largest singular value of $CK$, $C$ is a weight matrix, and $K$ the Kalman gain.
\end{theorem}

\begin{proof}
The result directly holds since from Lemma~\ref{lemma:q0-bound},
$|J(a) - \breve{J}(a)| = |(b_0 - \breve{b}_0)^T a|
    \leq \norm{b_0 - \breve{b}_0} \norm{a}\leq \\
    \leq \delta \sigma_{\text{max}}(CK) \norm{a}.
$
\end{proof}
\subsection{Design of MASA Attacks}
We now employ the same procedure to design MASA attacks and bound online attack performance. We start with the following result for fully optimal MASA attack design.
\begin{proposition}\label{prop:es-masa}
The estimated-state MASA attack (Def.~\ref{def:es-masa-attack}) is the solution of the optimization~problem
\begin{align} \label{eq:es-masa}
    \begin{split}
        a^*_{\text{masa, es}} &= \argmin_a \ \onehalf a^T A_0 a + \Bar{b}_0^T a \\
        & \text{subject to} \ \onehalf a^T A_1 a + b_1^T a + d_1 \leq 0.
    \end{split}
\end{align}
\end{proposition}
\begin{proof}
Follows directly from Proposition~\ref{prop:mmdopt} by replacing $Cx$ with $\mathcal{X}^a$ where $\mathcal{X}^a$ is the attacker specified state.
\end{proof}

Note that the estimated-state MASA attack does not require knowledge of the true state of the plant and is thus feasible (i.e., can be executed online).

\subsubsection{Relaxations}
The true-state MASA attack from Definition~\ref{def:es-masa-attack} requires the true plant state which is unavailable to the attacker. Furthermore, the relaxation using $\breve{x}$ instead of the true state is not sufficient in the MASA attack due to the delayed dependence of $\breve{x}$ on $a$ -- i.e., the attack impacts $\breve{x}$ after control corrects for errors in $\hat{x}$ (which \emph{does} depend on $a$ through the compromised update in~\eqref{eq:comp-ekf-update}). Such delayed dependence can be highly non-linear and depends on the victim's goal and controller which are not fully available under a full-reactive knowledge model.

Therefore, we propose an alternative relaxation using a reflection of the attacker's estimated state of the plant.

\begin{definition}[Reflected True-State MASA] \label{def:ref-as-masa-attack}
Let $\mathcal{X}^a\in\mathbb{R}^w$ be an attacker-specified state for true-state MASA (Def.~\ref{def:es-masa-attack}). Let $C$ be an attacker-specified projection matrix. An attack $a^\dagger$ is a reflected true-state MASA attack~if obtained as
\begin{align}
    \begin{split}
        a^\dagger_{\text{masa, ts}} = \argmin_a \ \| C \hat{x}(a) - \mathcal{X}^b \|_2^2\\
        \text{s.t.} \quad \Tilde{y}(a) S^{-1} \Tilde{y}^T(a) \leq \tau \\
        \mathcal{X}^b \coloneqq C\hat{x}^- + C\left( \breve{x} - \mathcal{X}^a \right).
    \end{split}
\end{align}
\end{definition}

Now, we can capture the following result.
\begin{proposition}\label{prop:ts-masa}
The reflected true-state MASA attack can be obtained as a solution to the following problem
\begin{align}\label{eq:ts-masa}
    \begin{split}
        a^\dagger_{\text{masa, ts}} &= \argmin_a \ \onehalf a^T A_0 a + \breve{\Bar{b}}_0^T a \\
        & \text{subject to} \ \onehalf a^T A_1 a + b_1^T a + d_1 \leq 0.
    \end{split}
\end{align}
\end{proposition}

\begin{proof}
Follows from Proposition~\ref{prop:mmdopt}, by  replacing $Cx$ with $\mathcal{X}^b$ "where $\mathcal{X}^b$ is the reflection of the attacker's goal state across the current state.
\end{proof}

The reflected true-state MASA attack will myopically push $\hat{x}$ in the \emph{opposite} direction of the \emph{attacker's} goal state, $\mathcal{X}^a$, which intuitively will cause the control to compensate \emph{towards} the adversary's goal state. However, without access to the control module, i.e.,~without knowing how the control will react to the state estimate error, this has few guarantees and the worst-case error may be difficult if not impossible to bound. That said, we find that it works well in practice.
\section{Evaluation} \label{sec:5-experiments}

We demonstrate impact of the myopic attacks on nonlinear state estimation in a kinematic case study. Subsequently, the bounds from Theorems~\ref{theorem:att-err-bound},~\ref{theorem:MMD-err-bound}, and~\ref{theorem:MMD-impact-bound}, and Lemma~\ref{lemma:q0-bound} are validated using Monte Carlo~(MC) simulations.

\subsection{Case Studies}

We use a nonlinear kinematic state estimation application. Still, the presented principles and experiments generalize to all applications of linear and extended Kalman~filtering.

\subsubsection{Model}
We simulate a dynamic target tracking scenario by modeling spherical coordinate returns from a radar sensor with component-wise Gaussian noise according to~\cite{Li2003}. We simulate range, azimuth, and elevation measurements $(\rho, \theta, \phi)$ relative to a fixed sensor platform and use an EKF to process measurements. We estimate position, velocity, and acceleration states using a nearly-constant-acceleration model from~\cite{Li2003}. We allow the filter to converge over $t=10.0~s$ before starting the attacks and running until $t=18.0~s$.

\subsubsection{Methods}
Each attack objective is convex with closed-form gradients and Hessians. We pre-condition following~\cite{boyd2004convex} using the Cholesky factorization of the inverse constraint Hessian to achieve faster convergence. This step is essential to obtaining real-time convergence, particularly in the case of order-of-magnitude scaling discrepancies between measurements (i.e., $\rho >> \theta,\ \phi$).

We choose a constrained trust-region optimization algorithm and find that the optimization runs faster than the simulation rate, easily keeping up  with real-time.

\subsubsection{Case Study I -- MMD}

Fig.~\ref{fig:mmd_case_study} shows results of the MMD attacks with the projection matrix set as $C_{i,j} = \begin{cases} 1 & i=j \\ 0 & i \neq j \end{cases}$, $w=3$, $n=9$. The attack quickly compromises the victim's (i.e., plant) state estimate, even with nonlinearities in E\&C. Fig.~\ref{fig:mmd_case_study} illustrates that the attack never exceeds the threshold set by the $\chi^2$ anomaly detector  meaning the attack remains stealthy, entirely in accordance with Proposition~\ref{prop:chi2-equality}.

\begin{figure}[t]
     \centering
     \begin{subfigure}[b]{\linewidth}
         \centering
         \includegraphics[width=0.76\textwidth]{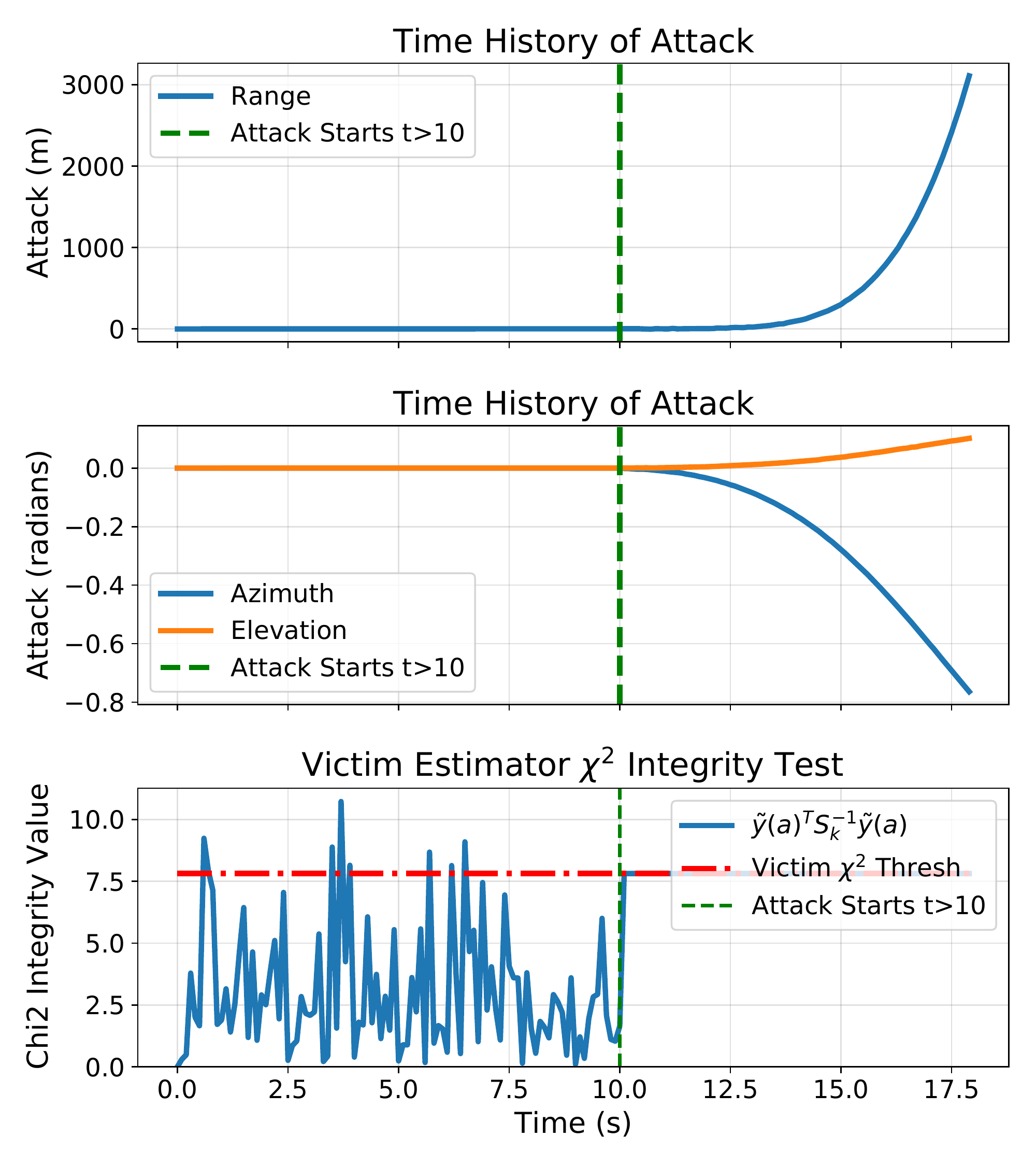}
         \vspace{-6pt}
     \end{subfigure}
    \caption{Nonlinear estimation is easily compromised by myopic MMD attacks. Attack quickly and stealthily compromises spherical coordinate measurements. Estimation errors 
    grow beyond the uncertainty bounds suggested by the EKF.}
    \label{fig:mmd_case_study}
    \vspace{-6pt}
\end{figure}

\subsubsection{Case Study II-- MASA}
Fig.~\ref{fig:masa_case_study_2} shows the same model with an estimated-state MASA attack following~\eqref{eq:masa-attack}. The attacker drives $\hat{x}$ towards a specified goal state, $\mathcal{X}^a$. In this kinematic application, we find that solely specifying attack goal as a position state (i.e.,~$C\in\mathbb{R}^{3\times 9}_{+}$) does succeed in rapidly pushing the state estimate towards the attacker goal, but that overshoot occurs. This is expected since the attack was formulated as a myopic optimization. Thus, we choose an attacker goal state that has both position and velocity. Specifically, $C_{i,j} = \begin{cases} 1 & i=j \\ 0 & i \neq j \end{cases}$, $w=6$, $n=9$. We observe the objective remains constant at 0 in Fig.~\ref{fig:masa_case_study_2} without overshoot.

\begin{figure}[t]
    \centering
    \begin{subfigure}[b]{\linewidth}
        \centering
        \includegraphics[width=0.76\textwidth]{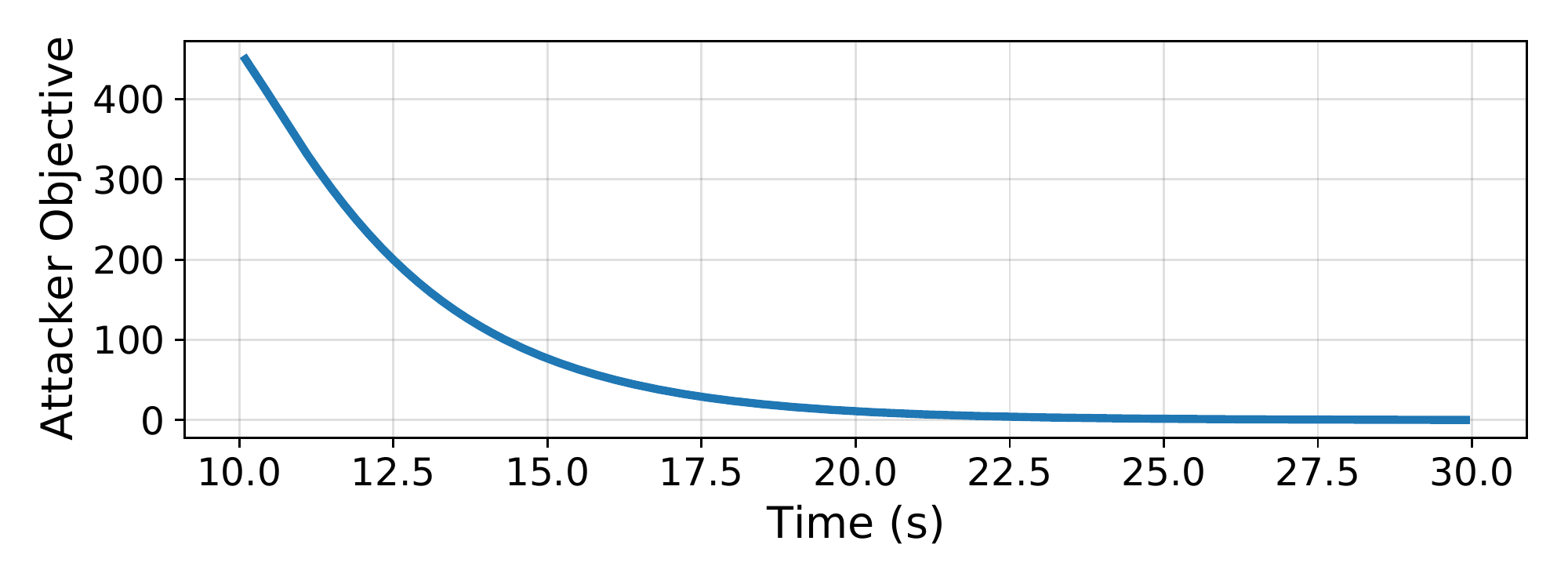}
        \vspace{-10pt}
    \end{subfigure}
    \begin{subfigure}[b]{\linewidth}
        \centering
        \includegraphics[width=0.76\textwidth]{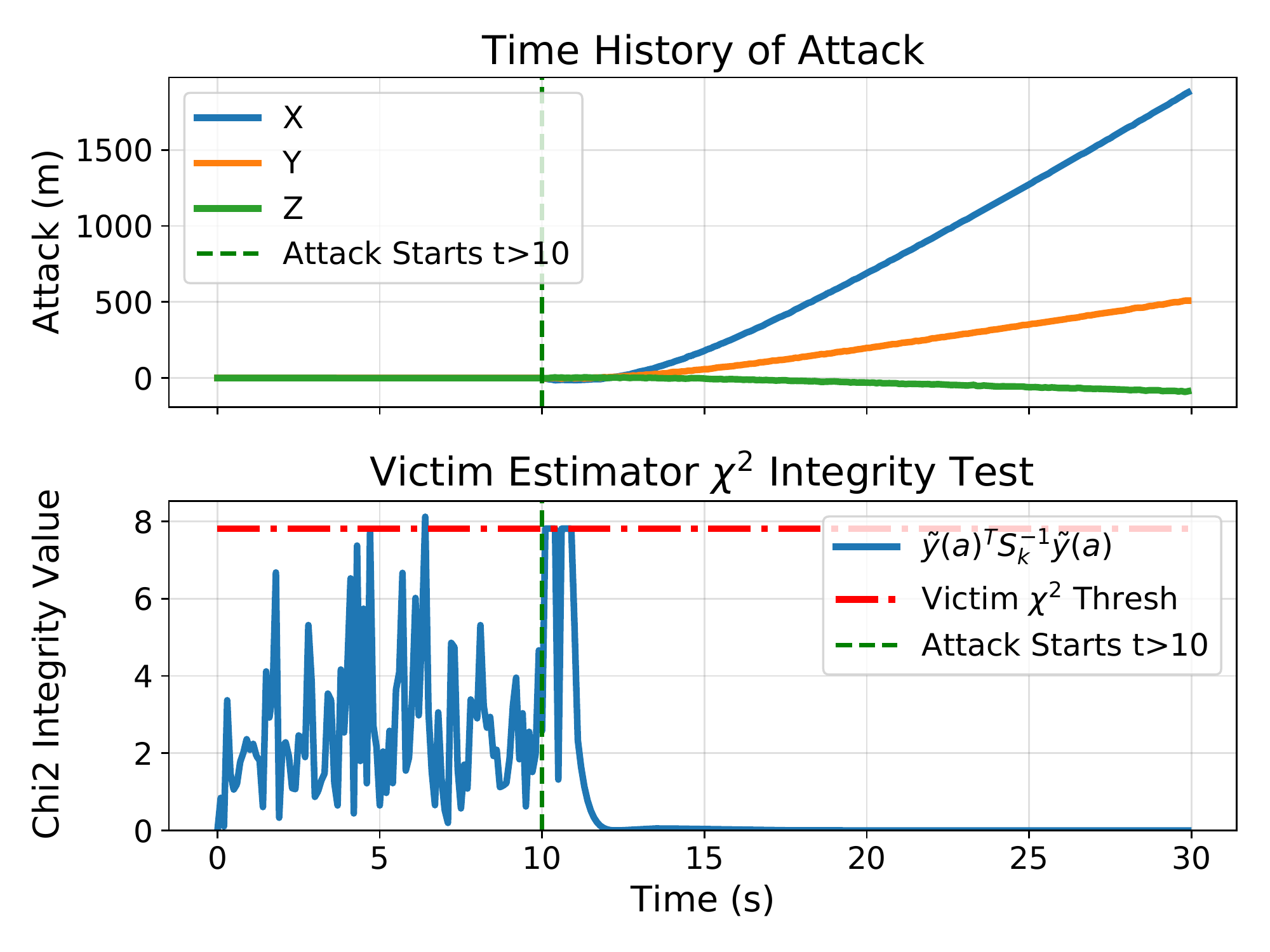}
        \vspace{-10pt}
    \end{subfigure}
    \caption{MASA attack quickly drives the estimated state towards attacker-specified state. The optimization objective drops to 0 once the attacker perfectly reaches the state.} 
    \label{fig:masa_case_study_2}
    \vspace{-10pt}
\end{figure}

\subsection{MC Bound Simulation}

Next, we use Monte Carlo simulations without a dedicated dynamics model to investigate the bounds derived in Sec.~\ref{sec:4-opt-attacks}. 

\subsubsection{Methods}
Given a fixed true state $x_k\in \mathbb{R}^9$, $\hat{x}_{k|k-1}$ and $\breve{x}_{k|k}$ are sampled from a Gaussian distribution given a fixed victim-state covariance matrix, $P\in\mathbb{R}^{9\times 9}_{+}$, $P\succ0$. A measurement model creates a measurement from the true state for the EKF. We choose $C\in\mathbb{R}^{3\times 9}_{+}$ to maximize the deviation in the first three states.

\subsubsection{Results}

$N=10000$ Monte Carlo trials are used to observe behavior of the myopic attacks. Fig.~\ref{fig:bounds_study} shows histograms of quantities derived in Theorems~\ref{theorem:att-err-bound},~\ref{theorem:MMD-err-bound},~\ref{theorem:MMD-impact-bound} and Lemma~\ref{lemma:q0-bound}. All bounds are order-of-magnitude tight.

\begin{figure}[!t]
     \centering
     \begin{subfigure}[b]{\linewidth}
         \centering
         \includegraphics[width=0.734\textwidth]{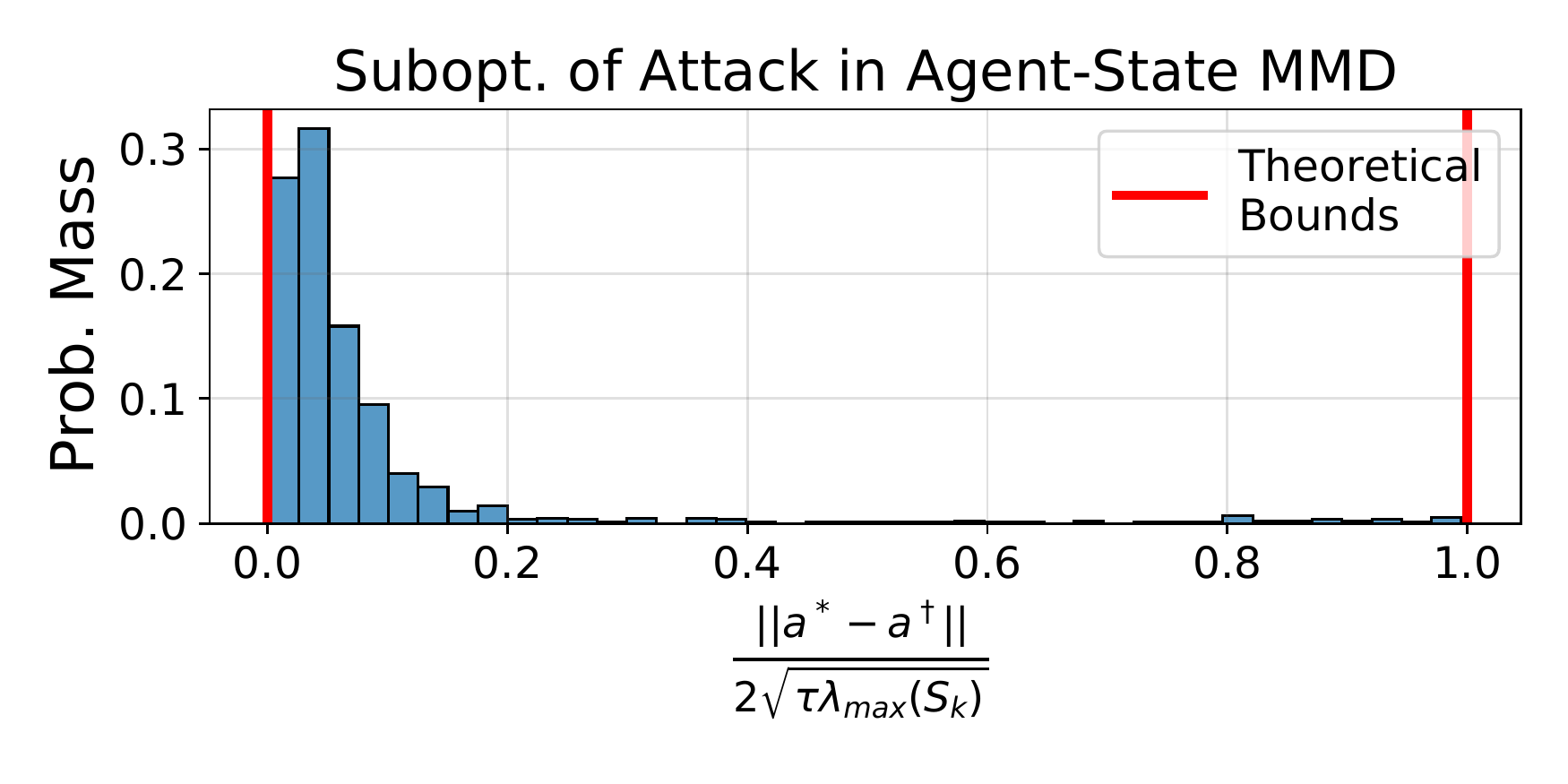}
         \vspace{-10pt}
         \label{fig:bound_attack_err}
     \end{subfigure}
     \begin{subfigure}[b]{\linewidth}
         \centering
         \includegraphics[width=0.734\textwidth]{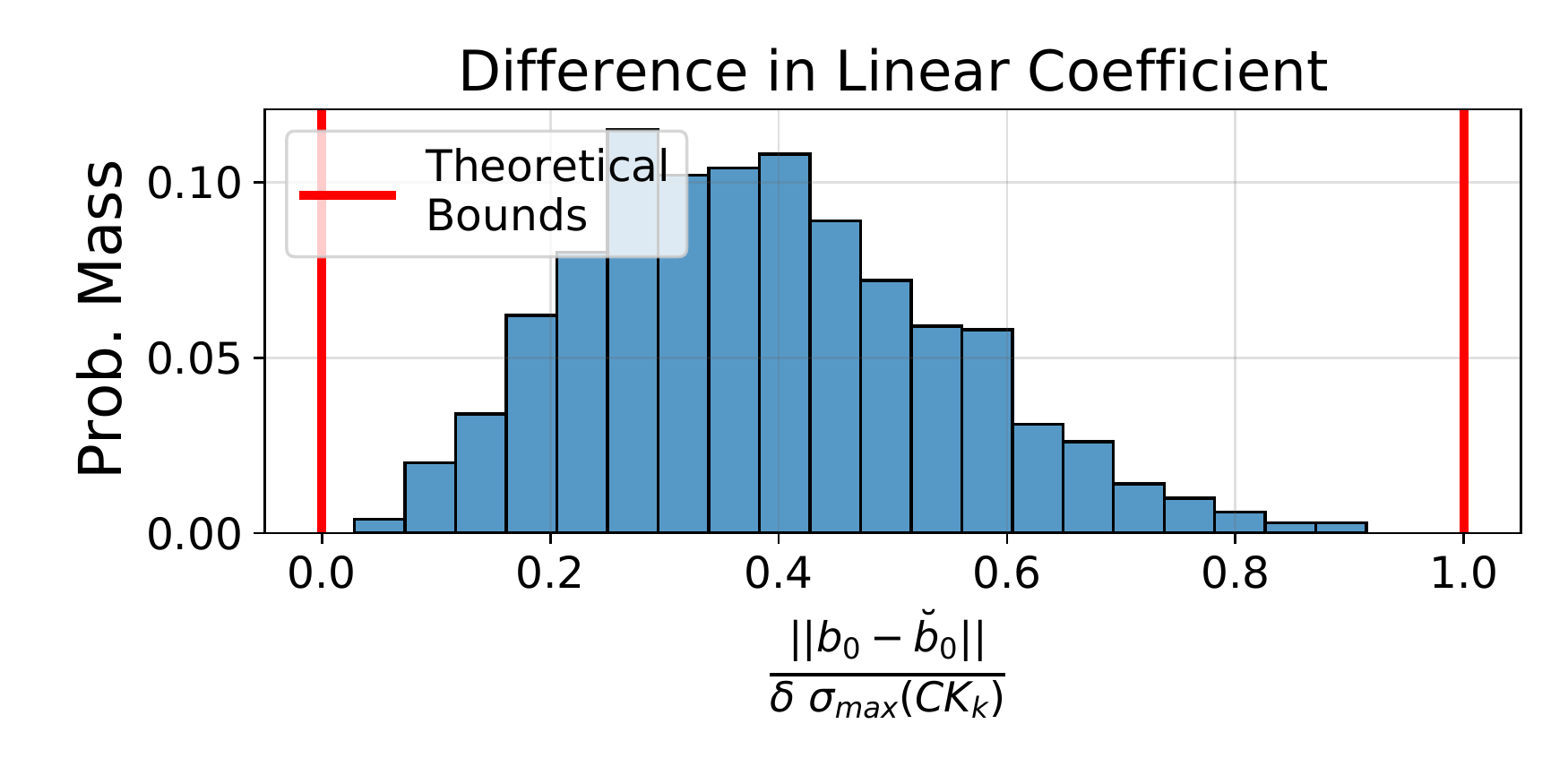}
         \vspace{-10pt}
         \label{fig:bound_linear_coeff}
     \end{subfigure}
     \begin{subfigure}[b]{\linewidth}
         \centering
         \includegraphics[width=0.734\textwidth]{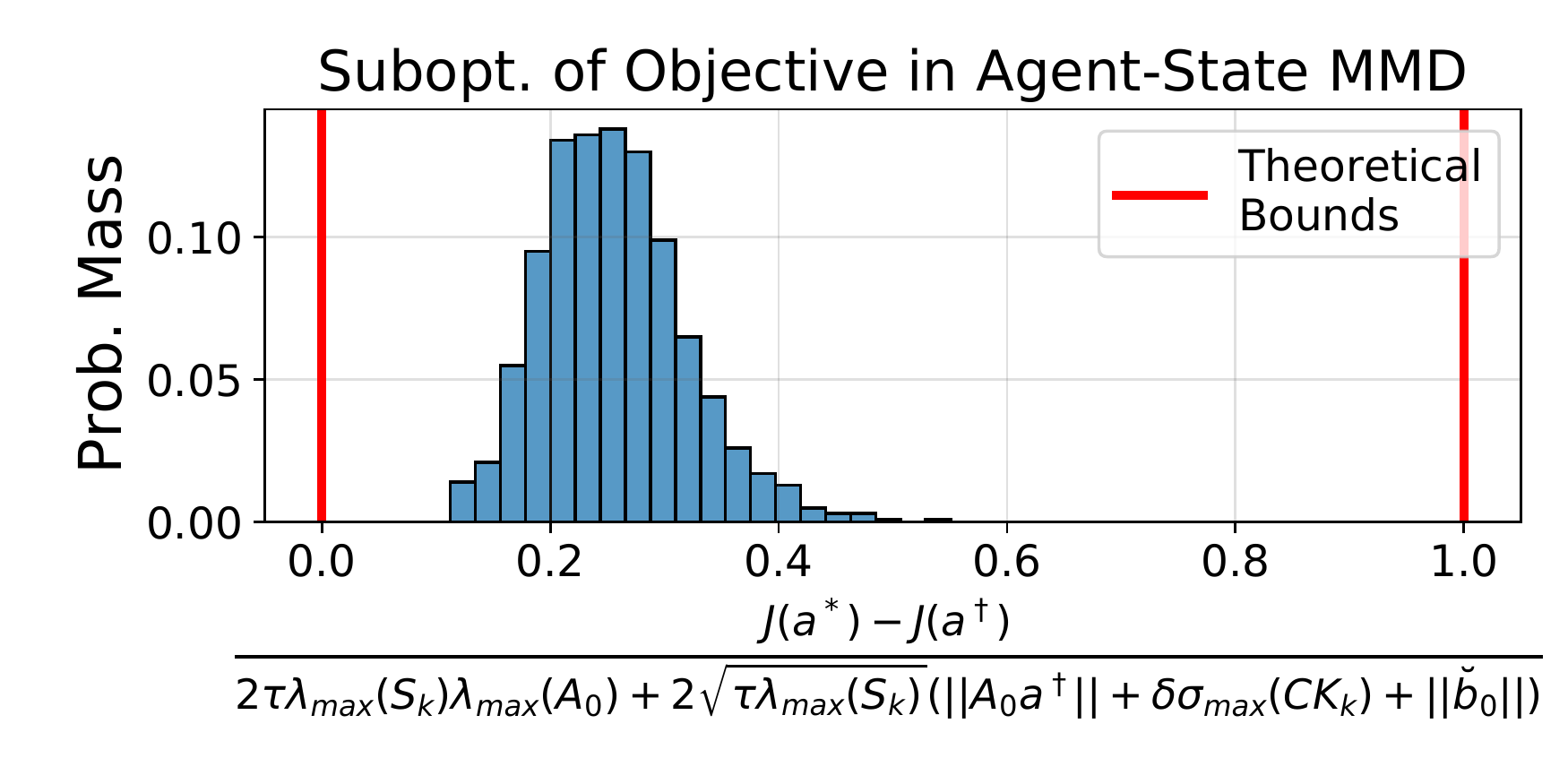}
         \vspace{-10pt}
         \label{fig:bound_subopt}
     \end{subfigure}
     \begin{subfigure}[b]{\linewidth}
         \centering
         \includegraphics[width=0.734\textwidth]{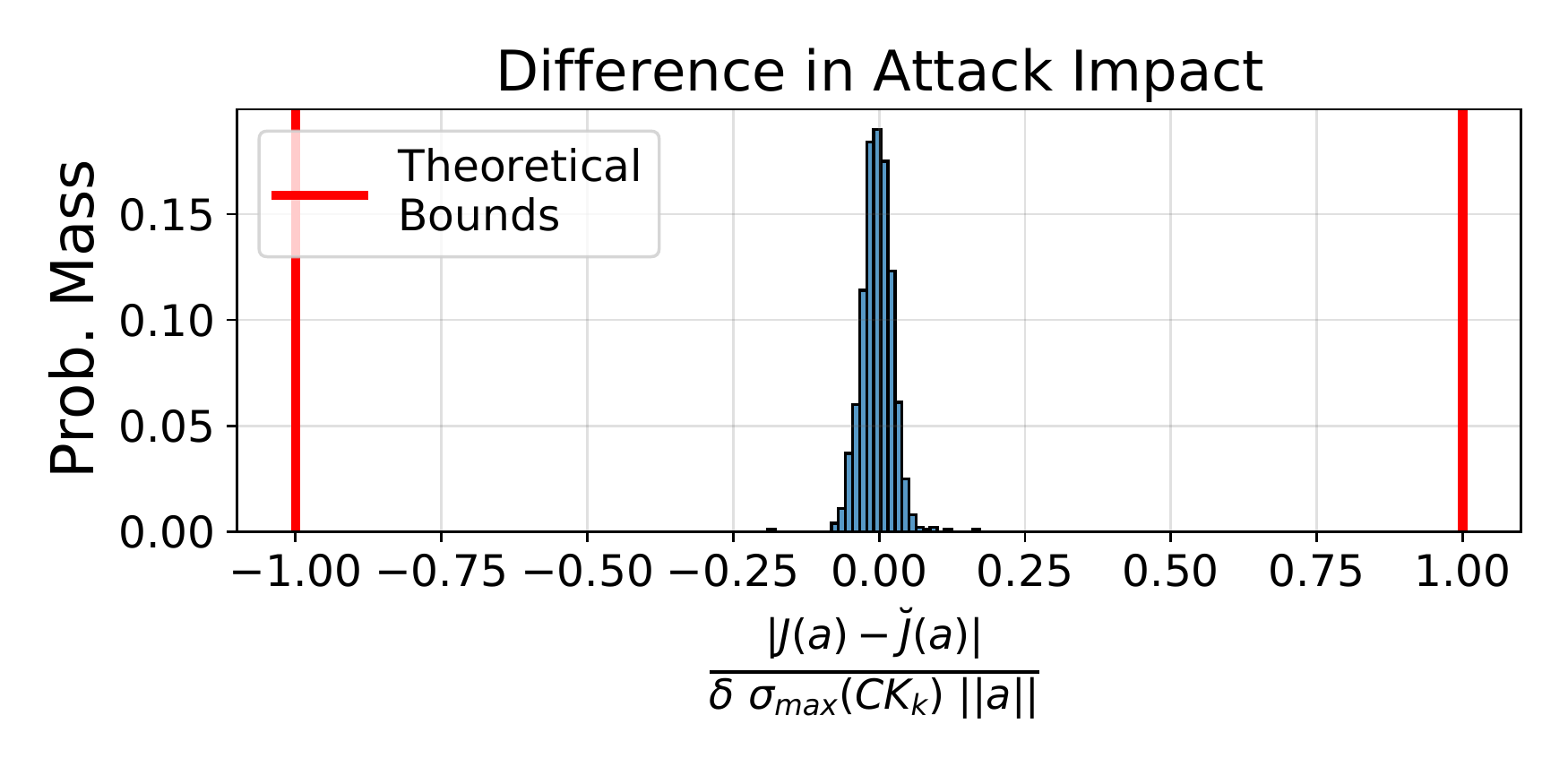}
         \vspace{-10pt}
         \label{fig:bound_att_impact}
     \end{subfigure}
    \caption{Derived bounds for MMD attack from $N=1000$ Monte Carlo simulations normalized by upper bound to fit between [\{-1, 0\}, 1]: Comparison with bounds from (a) Theorem~\ref{theorem:att-err-bound}, (b) Lemma~\ref{lemma:q0-bound}, (c) Theorem~\ref{theorem:MMD-err-bound}, and (d) Theorem~\ref{theorem:MMD-impact-bound}.}
    \label{fig:bounds_study}
    \vspace{-12pt}
\end{figure}

\section{Conclusion} \label{sec:7-conclusion}

We defined myopic maximum deviation and myopic adversarial state approach attacks. When attacking EKFs with a $\chi^2$ anomaly detector, each attacker goal can be formulated as a convex QCQP. We provided practical relaxations to ensure run-time feasibility given an appropriate attacker knowledge model. Finally, we showed that the difference between the optimal and relaxed problems is bounded. Future work will use this as a basis to derive attacks with relaxed information models and develop robust estimators for nonlinear~systems.


\bibliographystyle{IEEEtran}
\bibliography{references}

\end{document}